\def\year{2019}\relax
\documentclass[letterpaper]{article} 
\usepackage{aaai19}  
\usepackage{times}  
\usepackage{helvet}  
\usepackage{courier}  
\usepackage{url}  
\usepackage{graphicx}  
\usepackage{xcolor}
\usepackage{enumerate}
\usepackage{array}
\usepackage{amssymb}
\usepackage{amsmath}
\usepackage{mathrsfs}
\usepackage{amsthm}
\usepackage{newfloat}
\usepackage{mathtools}
\usepackage{amsmath}
\usepackage{empheq}
\usepackage{thmtools, thm-restate}
\usepackage{standalone}
\usepackage{tikz}
\usepackage{float}  

\newtheorem{theorem}{Theorem}
\newtheorem{proposition}{Proposition}
\newtheorem{corollary}{Corollary}
\newtheorem{observation}{Observation}

\newtheorem{remark}{Remark}
\newtheorem{definition}{Definition}

\newcommand{\citet}[1]{\citeauthor{#1}~\shortcite{#1}}
\newcommand{\citep}{\cite}
\newcommand{\cI}{\ensuremath{\mathcal{I}}}
\newcommand{\leafnodes}{\ensuremath{\mathcal{Z}}}

\DeclareFloatingEnvironment{lp}

\begin{document}
%
\title{Quasi-Perfect Stackelberg Equilibrium}
\author{
	Alberto Marchesi,\textsuperscript{\textnormal 1} Gabriele Farina,\textsuperscript{\textnormal 2} Christian Kroer,\textsuperscript{\textnormal 2} Nicola Gatti,\textsuperscript{\textnormal 1} Tuomas Sandholm\textsuperscript{\textnormal 2} \\
		\textsuperscript{1} Politecnico di Milano, Piazza Leonardo da Vinci 32, I-20133, Milan, Italy \\
	\textsuperscript{2} Carnegie Mellon University, 5000 Forbes Avenue, Pittsburgh, PA 15213, USA \\
	alberto.marchesi@polimi.it, gfarina@cs.cmu.edu, ckroer@cs.cmu.edu, \\
	nicola.gatti@polimi.it, sandholm@cs.cmu.edu 
}	

\maketitle

\begin{abstract}
  Equilibrium refinements are important in extensive-form (\emph{i.e.}, tree-form) games, where they
  amend weaknesses of the Nash equilibrium concept by requiring sequential rationality and other beneficial properties.
  One of the most attractive refinement concepts is \emph{quasi-perfect equilibrium}. While quasi-perfection has been studied in extensive-form games, it is poorly
  understood in Stackelberg settings---that is, settings where a leader can commit to a strategy---which are important for
  modeling, for example, security games. In this paper, we introduce the axiomatic
  definition of \emph{quasi-perfect Stackelberg equilibrium}. We develop a broad class of
  game perturbation schemes that lead to them in the limit. Our
  class of perturbation schemes strictly generalizes prior perturbation
  schemes introduced for the computation of (non-Stackelberg) quasi-perfect
  equilibria. Based on our perturbation schemes, we develop a branch-and-bound
  algorithm for computing a quasi-perfect Stackelberg equilibrium. It leverages a perturbed variant of the linear program for computing a
  Stackelberg extensive-form correlated equilibrium. Experiments show that our algorithm can be used to find an approximate quasi-perfect Stackelberg equilibrium in games with thousands of nodes.
\end{abstract}

\section{Introduction}\label{ec:introduction}

The main solution concept in game theory, \emph{Nash equilibrium (NE)}, may prescribe non-credible
strategies in \emph{extensive-form (i.e., tree-form) games (EFGs)}. To solve that problem, equilibrium refinements have been proposed for such games~\cite{selten1975reexamination}.  Among the
plethora of NE refinements
(see~\citeauthor{vanDamme:1987:SPN:38403}~\shortcite{vanDamme:1987:SPN:38403}
for details), the \emph{quasi-perfect equilibrium (QPE)}, proposed
by~\citeauthor{van1984relation}~\shortcite{van1984relation}, plays a central
role, and it is considered one of the most attractive NE refinement concepts, as
argued, for example,
by~\citeauthor{RePEc:eee:gamebe:v:8:y:1995:i:2:p:378-388}~\shortcite{RePEc:eee:gamebe:v:8:y:1995:i:2:p:378-388}.
The rationale behind the QPE concept is that every player, in every information
set, plays her best response to perturbed---that is, subject to
trembles---strategies of the opponents. Unlike the \emph{normal-form perfect
equilibrium}, the QPE guarantees that the strategies of the players are
sequentially rational, and furthermore, quasi-perfection implies normal-form
perfection. Unlike the \emph{extensive-form perfect equilibrium (EFPE)}, in a QPE
every player (reasonably) assumes that she will not make mistakes in the
future, and this excludes some unreasonable
strategies~\cite{RePEc:eee:gamebe:v:8:y:1995:i:2:p:378-388}. Computation of NE
refinements has received extensive attention in the literature. In the
two-player case,
\citeauthor{miltersen2010computing}~\shortcite{miltersen2010computing} provide
algorithms for finding a QPE, while~\citeauthor{DBLP:conf/aaai/Farina017}~\shortcite{DBLP:conf/aaai/Farina017}
for finding an EFPE. In particular,
\citeauthor{miltersen2010computing}~\shortcite{miltersen2010computing} show that
a strict subset of the QPEs can be found when the sequence form is subject to a
specific perturbation,
while~\citeauthor{DBLP:conf/aaai/Farina017}~\shortcite{DBLP:conf/aaai/Farina017}
do the same for the EFPE. Iterative algorithms for such perturbed games in the
zero-sum EFPE setting were introduced by \citet{kroer2017smoothing} and
\citet{farina2017regret}.\footnote{\emph{Normal-form proper equilibrium} is a refinement of
QPE~\citep{van1984relation}, but it has drawbacks:
(1) it requires players to assume a very specific structure on trembles
which is not necessarily well-motivated, (2) the minimum tremble magnitudes depend on the action probabilities, which begets additional computational challenges, and (3) it is unknown whether it can be represented via
perturbation schemes, even in the non-Stackelberg setting.
For the zero-sum
case, \citet{miltersen2008fast} show a polynomial-time approach using the
sequence form, but it is based on solving a large (possibly linear in game-size)
number of LPs, and thus may not be practical.
For the general-sum case, it is
not even known whether the sequence form can be applied; the only known approach
relies on conversion to normal form---which causes an exponential blow-up---and then applying a pivoting
algorithm~\citep{Sorensen12:Computing}.}

In \emph{Stackelberg games}, a \emph{leader} commits to a (possibly mixed) strategy first, and then a \emph{follower}
best responds to that strategy~\citep{Stackelberg34:Marktform}. Stackelberg games have received significant attention in recent
years~\citep{conitzer2006computing} due to their applications, for example, in security domains~\cite{tambe2011security}.

Work on equilibrium refinements in the context of \emph{Stackelberg
extensive-form games} has only started recently. Akin to usual extensive-form game refinements, \emph{Stackelberg
equilibrium (SE)} refinements should guarantee both the optimality of the commitment
off the equilibrium path and some form of robustness against small trembles of
the opponent.

To our knowledge, there is only one prior study of refinements for Stackelberg
extensive-form games~\citep{DBLP:conf/ijcai/FarinaMK0S18}. They
characterize a set of SE refinements based on what solutions can be obtained by imposing a perturbation scheme on the game---where players tremble onto suboptimal strategies with some small probabilities---and taking the limit as the trembling probability approaches zero.
They prove that, for any perturbation scheme, all the limit points of sequences of SEs in a perturbed game are SEs of the original, unperturbed game. Interestingly, they prove that when restricting attention to the common tie-breaking rules for the follower (\emph{strong} SE assumes the follower breaks ties in the best way for the leader and \emph{weak} SE assumes the follower breaks tie in the worst way for the leader), this is no longer the case.
Their approach does not start from a game-theoretic, axiomatic definition of the refinement concept. As we show in this paper, their approach captures only a strict subset of the solutions that are consistent with our natural game-theoretically defined refinement concept.
One way to view this is that their operational definition is deficient in that it does not characterize all the solutions that are consistent with the natural, axiomatic definition of the refinement concept. Another view is that they have an operational definition and we provide a generalization.

In terms of complexity, they prove that 
finding any SE
is
$\mathsf{NP}$-hard. (Hardness had previously been proven for finding a strong SE~\citep{Letchford:2010:COS:1807342.1807354}.)
Therefore, finding any SE refinement is also $\mathsf{NP}$-hard.

\textbf{Our contributions}. In this paper, we formally define the 
\emph{quasi-perfect Stackelberg equilibrium (QPSE)} refinement game theoretically in the same axiomatic fashion as QPE was defined for non-Stackelberg games~\citep{van1984relation}. As in the case of QPEs, our
definition is based on a set of properties of the players' strategies, and it cannot
be directly used to search for a QPSE. Subsequently, we define a class of
perturbation schemes for the sequence form such that any limit point of a sequence of SEs in a perturbed game is a QPSE. This class of perturbation schemes
strictly includes those used to find a QPE
by~\citeauthor{miltersen2010computing}~\shortcite{miltersen2010computing}.
%
%
%
Then, we extend the algorithm by~\citet{Cermak16:Using} to the case of QPSE
computation. We derive the corresponding mathematical program for computing a \emph{Stackelberg extensive-form correlated equilibrium (SEFCE)} when a
perturbation scheme is introduced and we discuss how the individual steps of the
algorithm change. In particular, the implementation of our algorithm is much
more involved, requiring the combination of branch-and-bound techniques with
arbitrary-precision arithmetic to deal with small perturbations. This does not
allow a direct application of off-the-shelf solvers. Finally, we experimentally
evaluate the scalability of our algorithm.

\section{Preliminaries}\label{sec:preliminaries}

Using standard notation~\citep{shoham2008multiagent}, a \emph{Stackelberg extensive-form game (SEFG)} of imperfect information is a tuple $(\mathcal{N},\mathcal{H},\mathcal{Z},\mathcal{A},\rho,\chi,\mathcal{C},u,\mathcal{I} )$.
$\mathcal{N} = \{\ell,f \}$ is the set of players, the leader and the follower.
$\mathcal{H} = \mathcal{H}_c \cup \mathcal{H}_\ell \cup \mathcal{H}_f$ is the set of nonterminal nodes, where $\mathcal{H}_c$ is the set of chance nodes, while $\mathcal{H}_\ell$ and $\mathcal{H}_f$ are the sets of leader's and follower's decision nodes, respectively.
$\mathcal{Z}$ is the set of terminal nodes.
$\mathcal{A} = \mathcal{A}_c \cup \mathcal{A}_\ell \cup \mathcal{A}_f $ is the set of actions, where $\mathcal{A}_c$ contains chance moves, while $\mathcal{A}_\ell$ and $\mathcal{A}_f$ are the sets of leader's and follower's actions, respectively.
$\rho: \mathcal{H} \rightarrow 2^\mathcal{A}$ is the action function that assigns to each nonterminal node a set of available actions.
$\chi: \mathcal{H} \times \mathcal{A} \rightarrow \mathcal{H} \cup \mathcal{Z}$ is the successor function that defines the node reached when an action is performed in a nonterminal node.
$\mathcal{C} : \mathcal{H} \cup \mathcal{Z} \rightarrow[0,1]$ assigns each node with its probability of being reached given chance moves.
$u=\{u_\ell, u_f\}$, where $u_\ell, u_f: \mathcal{Z} \rightarrow \mathbb{R}$ specify leader's and follower's payoffs, respectively, in each terminal node.
Finally, $\mathcal{I} = \{ \mathcal{I}_\ell, \mathcal{I}_f \}$, where $\mathcal{I}_\ell$ and $\mathcal{I}_f$ define partitions of $\mathcal{H}_\ell$ and $\mathcal{H}_f$, respectively, into information sets, that is, groups of nodes that are indistinguishable by the player.
For every information set $I \in \mathcal{I}_i$ and nodes $h, \hat{h} \in I$, it must be the case that $\rho(h) = \rho(\hat{h}) = A(I)$, otherwise player $i$ would be able to distinguish the two nodes.
%
%

We focus on \emph{perfect-recall} SEFGs in which no player forgets what she did or knew in the past, that is, for every $i \in \mathcal{N}$ and $I \in \mathcal{I}_i$, all nodes belonging to $I$ share the same player $i$'s moves on their paths from the root.
%
%
%
Thus, we can restrict the attention to \emph{behavioral strategies}~\cite{kuhn2016extensive}, which define, for every player $i \in \mathcal{N}$ and information set $I \in \mathcal{I}_i$, a probability distribution over the actions $A(I)$.
For $i \in \mathcal{N}$, let $\pi_i \in \Pi_i$ be a player $i$'s behavioral strategy, with $\pi_{i a}$ denoting the probability of playing action $a \in \mathcal{A}_i$.
%
%
%
Overloading notation, we use $u_i$ as if it were defined over strategies instead of terminal nodes. Specifically, $u_i(\pi_\ell, \pi_f)$ is player $i$'s expected utility when $\pi_\ell \in \Pi_\ell$ and $\pi_f \in \Pi_f$ are played.

Perfect-recall SEFGs admit an equivalent representation called the \emph{sequence form}~\cite{von1996efficient,Romanovskii62:Reduction}.
Every node $h \in \mathcal{H} \cup \mathcal{Z}$ defines a \emph{sequence} $\sigma_i(h)$ for player $i \in \mathcal{N}$, which is the ordered set of player $i$'s actions on the path from the root to $h$.
%
Let $\Sigma_i$ be the set of player $i$'s sequences.
As usual, let $\sigma_\emptyset \in \Sigma_i$ be a fictitious element representing the empty sequence.
%
%
In perfect-recall games, given an information set $I \in \mathcal{I}_i$, for any pair of nodes $h, \hat{h} \in I$ it holds $\sigma_{i}(h) = \sigma_{i}(\hat{h}) = \sigma_i(I) $.
%
%
Given $\sigma_i \in \Sigma_i$ and $a \in A(I)$ with $ I \in \mathcal{I}_i : \sigma_i = \sigma_i(I)$, we denote as $\sigma_i a$ the \emph{extended} sequence obtained by appending $a$ to $\sigma_i$.
Moreover, for any pair $\sigma_i, \hat{\sigma}_i \in \Sigma_i$, we write $\hat{\sigma}_i \sqsubseteq \sigma_i$ whenever $\hat{\sigma}_i$ is a \emph{prefix} of $\sigma_i$, that is, $\sigma_i$ can be obtained by extending $\hat{\sigma}_i$ with a finite number of actions.
Given $\sigma_i \in\Sigma_i$, we also let $I_i(\sigma_i)$ be the information set $I \in \mathcal{I}_i : \sigma_i = \sigma_i(I)a$ for some $a \in A(I)$.
In the sequence form, a strategy, called a \emph{realization plan}, assigns each sequence with its probability of being played.
For $i \in \mathcal{N}$, let $r_i \in R_i$ be a player $i$'s realization plan.
In order to be well-defined, a realization plan $r_i$ must be such that $r_i(\sigma_\emptyset)=1$ and, for $I \in \mathcal{I}_i$, $r_i(\sigma_i(I)) = \sum_{a \in A(I)} r_i(\sigma_i(I)a)$.
Finally, letting $\Sigma = \Sigma_\ell \times \Sigma_f$ be the set of sequence pairs $\sigma = (\sigma_\ell, \sigma_f)$, overloading notation, $u_i : \Sigma \rightarrow \mathbb{R}$ is player $i$'s utility function in the sequence form, with $u_i(\sigma) = \sum_{h \in \mathcal{Z} : \sigma_{\ell}(h) = \sigma_\ell \wedge \sigma_{f}(h) = \sigma_f } u_i(h) \mathcal{C}(h)$.
%
Moreover, we also use $u_i$ as if it were defined over realization plans.
Formally, $u_i(r_\ell, r_f) = \sum_{\sigma \in \Sigma} u_i(\sigma) r_\ell(\sigma_\ell) r_f(\sigma_f)$.

The sequence form is usually expressed with matrix notation as follows.
Player $i$'s utility function is a $|\Sigma_\ell| \times |\Sigma_f|$ matrix $U_i$ whose entries are the utilities $u_i(\sigma)$, for $\sigma \in \Sigma$.
Constraints defining $r_i \in R_i$ are expressed as $F_i r_i = f_i$, where: $F_i$ is a $( |\mathcal{I}_i|+1 ) \times |\Sigma_i|$ matrix, $f_i \in \mathbb{R}^{|\mathcal{I}_i|+1}$, and, overloading notation, $r_i \in \mathbb{R}^{|\Sigma_i|}$ is a vector representing $r_i$.
Specifically, introducing a fictitious information set $I_\emptyset$, the entry of $F_i$ indexed by $(I_\emptyset, \sigma_\emptyset)$ is 1, and, for $I \in \mathcal{I}_i$ and $\sigma_i \in \Sigma_i$, the entry indexed by $(I,\sigma_i)$ is $-1$ if $\sigma_i = \sigma_i(I)$, while it is $1$ if $\sigma_i = \sigma_i(I)a$ for some $a \in A(I)$.
$F_i$ is zero everywhere else.
Moreover, $f_i^T = (1 \; 0 \cdots 0)$.
%
%
%
%
%
%
%
%
Finally, given $r_\ell \in R_\ell$ and $r_f \in R_f$, we can write $u_i(r_\ell, r_f) = r_\ell^T  U_i r_f $.

In perfect-recall games, behavioral strategies and realization plans are equally expressive.
Given $r_i \in R_i$, we obtain an equivalent $\pi_i \in \Pi_i$ by setting, for all $I \in \mathcal{I}_i$ and $a \in A(I)$, $\pi_{i a} = \frac{r_i(\sigma_i(I)a)}{r_i(\sigma_i(I))}$ when $r_i(\sigma_i(I)) > 0$, while $\pi_{i a}$ can be any otherwise.
Similarly, $\pi_i \in \Pi_i$ has an equivalent $r_i \in R_i$ with $r_i(\sigma_i) = \prod_{a \in \sigma_i} \pi_{i a}$ for all $\sigma_i \in \Sigma_i$.\footnote{Here, the equivalence is in terms of probabilities that the strategies induce on terminal nodes, \emph{i.e.}, it is \emph{realization equivalence}.}


The solution concept associated with SEFGs is the SE.
%
%
An SEFG may have many SEs, depending on the leader's assumption on how the follower breaks ties among multiple best responses.
A leader's strategy is part of an SE if it is optimal for \emph{some} tie-breaking rule of the follower.
%
Letting $\mathsf{BR}_{\Gamma} (\pi_\ell) = \arg\max_{{\pi}_f \in \Pi_f} u_f(\pi_\ell, {\pi}_f)$ be the set of follower's best responses to $\pi_\ell \in \Pi_\ell$ in an SEFG $\Gamma$,
%
we have the following formal definition of SE.\footnote{In this paper, we define SEs following a characterization introduced by~\citeauthor{DBLP:conf/ijcai/FarinaMK0S18}~\shortcite{DBLP:conf/ijcai/FarinaMK0S18} (Lemma 2 in their paper).}

\begin{definition}\label{def:se}
	Given an SEFG $\Gamma$, $(\pi_\ell, \pi_f)$ is an \emph{SE} of $\Gamma$ if $\pi_f  \in \mathsf{BR}_{\Gamma} (\pi_\ell)$ and, for all $\hat{\pi}_\ell \in \Pi_\ell$, there exists $\hat{\pi}_f \in  \mathsf{BR}_{\Gamma} (\hat{\pi}_\ell)$ such that $u_\ell(\pi_\ell, \pi_f) \geq u_\ell(\hat{\pi}_\ell, \hat{\pi}_f)$.
\end{definition}

Many papers on SEs focus on \emph{strong SEs} (SSEs), which assume that the follower breaks ties in favor of the leader.
%

\begin{definition}\label{def:sse}
	Given an SEFG $\Gamma$, $(\pi_\ell, \pi_f)$ is an \emph{SSE} of $\Gamma$ if $\pi_f  \in \mathsf{BR}_{\Gamma} (\pi_\ell)$ and, for all $\hat{\pi}_\ell \in \Pi_\ell$ and $\hat{\pi}_f \in  \mathsf{BR}_{\Gamma} (\hat{\pi}_\ell)$, it holds $u_\ell(\pi_\ell, \pi_f) \geq u_\ell(\hat{\pi}_\ell, \hat{\pi}_f)$.
\end{definition}

Finally, SEs and SSEs can be defined analogously for SEFGs in sequence form (using the equivalence between behavioral strategies and realization plans).

\section{Definition of Quasi-Perfect Stackelberg Equilibrim}\label{sec:definitions}

In this section, we introduce QPSEs, which refine SEs in SEFGs
using an approach resembling that adopted by~\citeauthor{van1984relation}~\shortcite{van1984relation} for defining QPEs in EFGs.
%
%


%
First, we provide needed additional notation.
We say that $\pi_i \in \Pi_i$ is \emph{completely mixed} if $\pi_{ia} > 0$ for all $a \in \mathcal{A}_i$.
%
%
Given two information sets $I, \hat{I} \in \mathcal{I}_i$, we write $I \succeq \hat{I}$ whenever $\hat{I}$ follows $I$, \emph{i.e.}, there exists a path from $h \in I$ to $\hat{h} \in \hat{I}$.
We assume $I_\emptyset \succeq \hat{I}$ for all $\hat{I} \in \mathcal{I}_i$ such that there is no $I \neq \hat{I} \in \mathcal{I}_i : I \succeq \hat{I}$.
In perfect-recall games, $\succeq$ is a partial order over $\mathcal{I}_i \cup \{ I_\emptyset \}$.
%
%
%
%
%
%
%
%
Given $\pi_i, \hat{\pi}_i \in \Pi_i$ and $I \in \mathcal{I}_i \cup \{I_\emptyset\}$, $\pi_i \big/_{I} \hat{\pi}_i$ is equal to $\hat{\pi}_i$ at all $\hat{I} \in \mathcal{I}_i : I \succeq \hat{I}$, while it is equal to $\pi_i$ everywhere else.
Moreover, for $I \in \mathcal{I}_i$, we write $\pi_i =_{I} \hat{\pi}_i$ if $\pi_{i a} = \hat{\pi}_{i a}$ for all $a \in A(I)$.
%
%
Finally, given completely mixed strategies $\pi_\ell \in \Pi_\ell$, $\pi_f \in \Pi_f$ and $I \in \mathcal{I}_i$, $u_{i, I}(\pi_\ell,\pi_f)$ denotes player $i$'s expected utility given that $I$ has been reached and strategies $\pi_\ell$ and $\pi_f$ are played.
%
%


Next, we introduce a fundamental building block:
the idea of follower's best response at an information set $I \in \mathcal{I}_f$.
Intuitively, $\pi_f$ is an $I$-best response to $\pi_\ell$ whenever playing as prescribed by $\pi_f$ at the information set $I$ is part of some follower's best response to $\pi_\ell$ in the game following $I$, given that $I$ has been reached during play. 
Formally:
\begin{definition}\label{def:info_set_br}
	Given an SEFG $\Gamma$, a completely mixed $\pi_\ell \in \Pi_\ell$, and $I \in \mathcal{I}_f$, we say that $\pi_f \in \Pi_f$ is an \emph{$I$-best response} to $\pi_\ell$, written $\pi_f \in \mathsf{BR}_{I}(\pi_\ell)$, if the following holds:
	$$
		 \max_{\substack{ \hat{\pi}_f \in \Pi_f : \\  \pi_f =_{I} \hat{\pi}_f }}  u_{f, I} \left( \pi_\ell, \pi_f \big/_{I} \hat{\pi}_f \right)  =  \max_{\hat{\pi}_f \in \Pi_f} u_{f, I} \left( \pi_\ell, \pi_f \big/_{I} \hat{\pi}_f  \right).
	$$
\end{definition}

For $i \in \mathcal{N}$ and $\pi_i \in \Pi_i$, let $ \{ \pi_{i, k} \}_{k \in \mathbb{N}}$ be a sequence of completely mixed player $i$'s strategies with $\pi_i$ as a limit point.
We are now ready to define the refinement concept.
In words, in a QPSE, the leader selects an optimal strategy to commit to in \emph{all} information sets, given that the follower best responds to it at \emph{every} information set, following \emph{some} tie-breaking rule.
Specifically,
%
point (ii) in Definition~\ref{def:qpse} ensures that the leader's commitment is optimal also in those information sets that are unreachable in absence of players' errors.
Notice that the leader only accounts for follower's future errors,
%
while the follower assumes that only the leader can make mistakes in future.
This is in line with the idea underlying QPEs in EFGs~\cite{van1984relation}.\footnote{\citet{van1984relation} defines a QPE of an $n$-player extensive-form game as a strategy profile $(\pi_i)_{i \in \mathcal{N}}$ obtained as a limit point of a sequence of completely mixed strategy profiles $\{ (\pi_{i,k})_{i \in \mathcal{N}} \}_{k \in \mathbb{N}}$ such that $\pi_i \in \mathsf{BR}_{I}((\pi_{j,k})_{j \neq i \in \mathcal{N}})$ for all $i \in \mathcal{N}$ and $I \in \mathcal{I}_i$.}
\begin{definition}\label{def:qpse}
	Given an SEFG $\Gamma$, $(\pi_\ell, \pi_f) $ is a \emph{quasi-perfect Stackelberg equilibrium (QPSE)} of $\Gamma$ if there exist sequences $ \{ \pi_{i, k} \}_{k \in \mathbb{N}}$, defined for every $i \in \mathcal{N}$ and $\pi_i \in \Pi_i$, such that:
	\begin{enumerate}[(i)]
		\item $\pi_f \in \mathsf{BR}_{I}(\pi_{\ell,k}) $ for all $ I \in \mathcal{I}_f$;
		\item for all ${I} \in \mathcal{I}_\ell \cup \{ I_\emptyset \}$ and $\hat{\pi}_\ell \in \Pi_\ell$, there exists $ \hat{\pi}_f \in \Pi_f : \hat{\pi}_f \in \mathsf{BR}_{\hat{I}}(\pi_{\ell,k} \big/_{I} \hat{\pi}_{\ell,k}) $ for all $ \hat{I} \in \mathcal{I}_f$, with:
		\begin{align}\label{eq:qpse}
		u_{\ell} \left(\pi_{\ell,k} \big/_{I} \pi_\ell, \pi_{f,k} \right) \geq u_{\ell} \left(\pi_{\ell,k} \big/_{I} \hat{\pi}_\ell, \hat{\pi}_{f,k} \right).
		\end{align}
	\end{enumerate}
\end{definition}

As with SEs, we introduce the \emph{strong} version of QPSEs.\footnote{Since Equation~\eqref{eq:qpse} must hold for every $\hat{\pi}_\ell \in \Pi_\ell$ and $ \hat{\pi}_f \in \Pi_f : \hat{\pi}_f \in \mathsf{BR}_{\hat{I}}(\pi_{\ell,k} \big/_{I} \hat{\pi}_{\ell,k}) $ for all $ \hat{I} \in \mathcal{I}_f$, Definition~\ref{def:qpsse} assumes that the follower breaks ties in favor of the leader.}

\begin{definition}\label{def:qpsse}
	Given an SEFG $\Gamma$, $(\pi_\ell, \pi_f) $ is a \emph{quasi-perfect strong Stackelberg equilibrium (QPSSE)} of $\Gamma$ if there exist $ \{ \pi_{i, k} \}_{k \in \mathbb{N}}$, defined for every $i \in \mathcal{N}$ and $\pi_i \in \Pi_i$, such that:
	\begin{enumerate}[(i)]
		\item $\pi_f \in \mathsf{BR}_{I}(\pi_{\ell,k}) $ for all $ I \in \mathcal{I}_f$;
		\item for all ${I} \in \mathcal{I}_\ell \cup \{ I_\emptyset \}$, $\hat{\pi}_\ell \in \Pi_\ell$, and $ \hat{\pi}_f \in \Pi_f : \hat{\pi}_f \in \mathsf{BR}_{\hat{I}}(\pi_{\ell,k} \big/_{I} \hat{\pi}_{\ell,k}) $ for all $ \hat{I} \in \mathcal{I}_f$, Equation~\eqref{eq:qpse} holds.
	\end{enumerate}
\end{definition}

As we will show in Section~\ref{sec:perturbed_games}, QPSEs are refinements of SEs, that is, any QPSE is also an SE.

\section{Family of Perturbation Schemes for QPSE}\label{sec:perturbed_games}

We now introduce a family of \emph{perturbation schemes} for SEFGs in sequence form that satisfies the following fundamental property: \emph{limits of SEs in perturbed sequence-form SEFGs are QPSEs of the original unperturbed SEFGs as the magnitude of the perturbation goes to zero}.
In addition to being theoretically relevant, this result enables us to design an algorithm for computing QPSEs in SEFGs (Section~\ref{sec:sefce}).

%
%

%
%
%
%
%

\begin{definition}[$\xi$-perturbation scheme]\label{def:qp_pert}
Given an SEFG $\Gamma$ and $i \in \mathcal{N}$, let $\xi_i : (0,1] \times Q_i \to \mathbb{R}^+$ be a function that maps a perturbation magnitude $\epsilon \in (0,1]$ and a sequence $\sigma_i \in \Sigma_i$ to a positive lower-bound $\xi_i(\epsilon,\sigma_i)$ on the probability of playing $\sigma_i$ such that:
	\begin{enumerate}[(i)]
		\item $\xi_i(\epsilon,\sigma_i)$ is a polynomial in $\epsilon$, for all $\sigma_i \in \Sigma_i$;
		\item $\lim_{\epsilon \rightarrow 0^+} \xi_i(\epsilon,\sigma_i) = 0$, for all $\sigma_i \in \Sigma_i \setminus \{\sigma_\emptyset\}$;
		\item $\lim_{\epsilon \rightarrow 0^+} \frac{ \xi_i(\epsilon, \sigma_i(I) a) }{ \xi_i(\epsilon, \sigma_i(I)) } = 0$, for all $I \in \mathcal{I}_i, a \in A(I)$.
	\end{enumerate}
	Then, a \emph{$\xi_{i}$-perturbation scheme} for $R_i$ is a function $\epsilon \mapsto R_i(\epsilon)$ defined over $\epsilon \in (0,1]$ in which $R_i(\epsilon)$ is the set of all $r_i \in R_i$ such that $r_i(\sigma_i) \geq \xi_i(\epsilon,\sigma_i)$ for all $\sigma_i \in \Sigma_i$.
\end{definition}
In words, the lower-bounds on sequence probabilities enjoy the following properties: (i) they are polynomials in the variable $\epsilon$; (ii) they approach zero as $\epsilon$ goes to zero; and (iii) $\xi_i(\epsilon,\sigma_i(I) a)$ approaches zero faster than $\xi_i(\epsilon,\sigma_i(I))$.

We denote by $(\Gamma, \xi_\ell, \xi_f)$ a \emph{$\xi$-perturbed} SEFG with $\xi_i$-perturbation schemes.
We let $\Gamma(\epsilon)$ be a particular sequence-form \emph{$\xi$-perturbed game instance} obtained from $\Gamma$ by restricting each set of realization plans $R_i$ to be $R_i(\epsilon)$.
We denote by $r_i(\epsilon)$ any realization plan in $R_i(\epsilon)$, and we let $\xi_i(\epsilon) \in \mathbb{R}^{|Q_i|}$ be a vector whose components are the lower-bounds $\xi_i(\epsilon,\sigma_i)$.
We denote by $\tilde{r}_i(\epsilon) = r_i(\epsilon) -\xi_i(\epsilon)$ the \emph{residual} of $r_i(\epsilon)$, which represents the part of player $i$'s strategy that is not fixed by the perturbation.\footnote{We assume without loss of generality that $\Gamma(\epsilon)$ is well-defined, that is, each set $R_i(\epsilon)$ is non-empty for every $\epsilon \in (0,1]$.}

Next, we state our main result about sequences of SEs in $\xi$-perturbed games. We postpone the proof to Section~\ref{sec:limits_sse}.

\begin{theorem}\label{thm:limit_se}
	Given a $\xi$-perturbed SEFG $(\Gamma, \xi_\ell, \xi_f)$, let $\{ \epsilon_k \}_{k \in \mathbb{N}} \rightarrow 0$ and let $\{ (r_\ell(\epsilon_k), r_f(\epsilon_k)) \}_{k \in \mathbb{N}}$ be a sequence of SEs in $\Gamma(\epsilon_k)$.
	Then, any limit point $(\pi_\ell, \pi_f)$ of the sequence $\{ (\pi_{\ell, k}, \pi_{f, k}) \}_{k \in \mathbb{N}}$ is a QPSE of $\Gamma$,
	where $(\pi_{\ell, k}, \pi_{f, k})$ are equivalent to $(r_\ell(\epsilon_k), r_f(\epsilon_k))$ for all $k \in \mathbb{N}$.
\end{theorem}

Theorem~\ref{thm:limit_se} also allows us to conclude the following, as a consequence of Theorem~1 of~\citeauthor{DBLP:conf/ijcai/FarinaMK0S18}~\shortcite{DBLP:conf/ijcai/FarinaMK0S18}.
\begin{corollary}\label{cor:qp_ref}
	Any QPSE of an SEFG $\Gamma$ is an SE of $\Gamma$.
\end{corollary}

Reuirements (ii)-(iii) in Definition~\ref{def:qp_pert} cannot be removed:
\begin{observation}\label{obv:obs_pert}
There are $\xi$-perturbed SEFGs $(\Gamma, \xi_\ell, \xi_f)$ with $\xi_i$-perturbation schemes that violate point (ii) or (iii) in Definition~\ref{def:qp_pert} for which Theorem~\ref{thm:limit_se} does not hold.
\end{observation}
\begin{proof}
Consider the SEFG in Figure~\ref{fig:examples}b with $\xi_\ell(\epsilon,a_\ell^1)=\xi_\ell(\epsilon,a_\ell^2)=\epsilon$ and $\xi_\ell(\epsilon,a_\ell^2a_\ell^3)= \xi_\ell(\epsilon,a_\ell^2a_\ell^4)= \frac{\epsilon}{3}$, which violates requirement (iii) in Definition~\ref{def:qp_pert}.
Clearly, any SE of $\Gamma(\epsilon)$ requires $r_\ell(\epsilon,a_\ell^1) = 1 - \epsilon$, $r_\ell(\epsilon,a_\ell^2) = \epsilon$, $r_\ell(\epsilon,a_\ell^2a_\ell^3) = \frac{\epsilon}{3}$, and $r_\ell(\epsilon, a_\ell^2 a_\ell^4) = \frac{2 \epsilon}{3}$.
Thus, any limit point of a sequence of SEs has $\pi_\ell(a_\ell^3) = \frac{1}{3}$ and $\pi_\ell(a_\ell^4) = \frac{2}{3}$, which cannot be the case in a QPSE of $\Gamma$, as the leader's optimal strategy at $\ell.2$ is to play $a_\ell^4$.
As for requirement (ii), we can build a similar example by setting $\xi_\ell(\epsilon,a_\ell^2)=\frac{1}{3}$.
\end{proof}

\begin{figure}[!h]
	\centering
	\includegraphics[scale=.8]{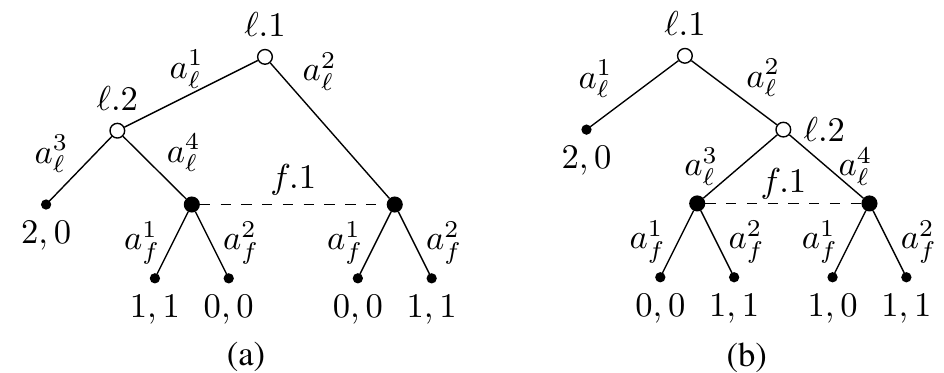}
	\caption{\small Examples SEFGs.}
	\label{fig:examples}
\end{figure}

\citeauthor{miltersen2010computing}~\shortcite{miltersen2010computing} introduced the idea of perturbing sequence-form EFGs in order to find a QPE. Our perturbation scheme generalizes theirs, where $\xi_i(\epsilon, \sigma_i) = \epsilon^{|\sigma_i|}$ for all $\sigma_i \in \Sigma_i \setminus \{\sigma_\emptyset\}$, with $|\sigma_i|$ being the number of actions in $\sigma_i$. There are games where our perturbation captures QPSEs that are \emph{not} obtainable with theirs. For instance, in the SEFG in Figure~\ref{fig:examples}a, $(\pi_\ell, \pi_f)$, with $\pi_\ell(a_\ell^1)=\pi_\ell(a_\ell^3)=1$, $\pi_\ell(a_\ell^2)=\pi_\ell(a_\ell^4)=0$, and $\pi_f(a_f^1)= \pi_f(a_f^2) = \frac{1}{2}$, is a QPSE that cannot be obtained with their perturbation scheme while it is reachable by setting $\xi_\ell(\epsilon, a_\ell^2)=\epsilon^2$. We observe that $(\pi_\ell, \pi_f)$ is also a QPE when we look at the game as an EFG without commitment; this shows that our perturbation scheme generalizes theirs also for QPEs.

Finally, when restricting the attention to SSEs, we can state the following: limits of SSEs are QPSSEs.
We make this formal in Theorem~\ref{thm:limit_sse} in the Appendix.

\section{Best Responses in $\xi$-Perturbed Games}\label{sec:props_games}

We now study properties of the follower's best responses to the leader's strategy in $\xi$-perturbed games. These properties will be useful for proving our results later in the paper.

In the following, letting $\Sigma_i(a) = \{ \sigma_i \in \Sigma_i \mid  a \in \sigma_i  \}$ for all $a \in \mathcal{A}_i$, $\Sigma_i(I) = \bigcup_{a \in A(I)} \Sigma_i(a)$ denotes player $i$'s sequences that pass through  information set $I \in \mathcal{I}_i$.
For ease of presentation, given $I \in \mathcal{I}_i$, $g_{i, I}(r_\ell,r_f) = \sum_{\sigma \in \Sigma : \sigma_i \in \Sigma_i(I)} u_i(\sigma) r_\ell(\sigma_\ell) r_f(\sigma_f)$ denotes player $i$'s expected utility contribution from terminal nodes reachable from $I$.
%
%
%
%
Finally, for $I \in \mathcal{I}_i$, let $R_i(I) \subseteq R_i$ be the set of $r_i \in R_i : r_i(\sigma_i(I))=1$, while, for $a \in A(I)$, $R_i(a) \subseteq R_i(I)$ is the set of $r_i \in R_i : r_i(\sigma_i(I)a)=1$.

%
%


Let $\mathsf{BR}_{\Gamma(\epsilon)}(r_\ell(\epsilon)) = \arg \max_{{r}_f(\epsilon) \in R_f(\epsilon)} u_f(r_\ell(\epsilon),{r}_f(\epsilon))$ be the set of follower's best responses to $r_\ell(\epsilon) \in R_\ell(\epsilon)$ in $\Gamma(\epsilon)$.
The next lemma gives a mathematical programming formulation of the follower's best-response problem in $\Gamma(\epsilon)$.
%
%
%

\begin{restatable}{lemma}{lemmanepert}\label{lem:ne_pert_game}
For every $r_\ell(\epsilon) \in R_\ell(\epsilon)$, $r_f(\epsilon) \in \mathsf{BR}_{\Gamma(\epsilon)}(r_\ell(\epsilon))$ if and only if $\tilde{r}_f(\epsilon)$ is optimal for Problem~$\mathcal{P}(\epsilon)$ below.
	%
	\begin{equation*}\label{prob:primal_tilde}
		\mathcal{P}(\epsilon)\ :\ \left\{\begin{aligned}
		\displaystyle\max_{\tilde{r}_f} &\quad r_\ell(\epsilon)^T  U_f \tilde{r}_f \\[-2mm]
		\textnormal{s.t.} &\quad F_f \tilde{r}_f = f_f  -  F_f \xi_f(\epsilon), \quad  \tilde{r}_f \geq 0.
		\end{aligned}\right.
	\end{equation*}
\end{restatable}
All omitted proofs are in the Appendix.

The dual of Problem~$\mathcal{P}(\epsilon)$ above is as follows.
\begin{proposition}\label{prop:dual_br_pert}
	For $r_\ell(\epsilon) \in R_\ell(\epsilon)$, Problem~$\mathcal{D}(\epsilon)$ below is the dual of Problem~$\mathcal{P}(\epsilon)$, where $v_f \in \mathbb{R}^{|\mathcal{I}_f| + 1}$ is the vector of dual variables.
	%
	\begin{subequations}\label{prob:dual_br_pert}
		\begin{empheq}[left={\mathcal{D}(\epsilon)\ :\ \empheqlbrace}]{align}
		\min_{v_f} &\quad \left( f_f - F_f \xi_f(\epsilon) \right)^T v_f \\[-2mm]
		\textnormal{s.t.} &\quad F_f^T v_f \geq U_f^T  r_\ell(\epsilon). \label{cons:dual}
		\end{empheq}
	\end{subequations}
	%
\end{proposition}


\begin{remark}
	Constraints~\eqref{cons:dual} in Problem~$\mathcal{D}(\epsilon)$ defined above ensure that, for every $I \in \mathcal{I}_f$ and $a \in A(I)$, we have
	\begin{align}\label{eq:constraints_rewritten}
		v_{f, I} \geq \hspace{-0.2cm} \sum_{\substack{\sigma \in \Sigma : \sigma_f = \sigma_f(I) a }} \hspace{-0.3cm} u_f(\sigma) r_\ell(\epsilon,\sigma_\ell)  +\hspace{-0.2cm} \sum_{\substack{ \hat{I} \in \mathcal{I}_f:  \sigma_f(\hat{I}) = \sigma_f(I)a}} \hspace{-0.5cm}v_{f, \hat{I}}.
	\end{align}	
\end{remark}

The optimal solutions to Problem~$\mathcal{D}(\epsilon)$ enjoy important properties that are stated in the following lemmas.
The first one says that, in an optimal solution, each variable $v_{f, I}$ must equal the maximum possible expected utility the follower can achieve following information set $I \in \mathcal{I}_f$.
The second lemma says that if an optimal solution to Problem~$\mathcal{D}(\epsilon)$ satisfies Constraint~\eqref{eq:constraints_rewritten} with equality for an information set $I \in \mathcal{I}_f$ and an action $a \in A(I)$, then playing $a$ at $I$ is optimal in the game following $I$.

\begin{restatable}{lemma}{lemmaperta}\label{lem:optimal_dual}
	For every $r_\ell(\epsilon) \in R_\ell(\epsilon)$, if $v_f^\ast \in \mathbb{R}^{|\mathcal{I}_f| + 1}$ is optimal for Problem~$\mathcal{D}(\epsilon)$, then for every $I \in \mathcal{I}_f$:
	\begin{align}\label{eq:inductive_argument}
		v_{f, I}^\ast = \max_{\substack{\hat{r}_f \in R_f(I) }} g_{f, I}( r_\ell(\epsilon), \hat{r}_f).
	\end{align}
\end{restatable}

\begin{restatable}{lemma}{lemmapertb}\label{lem:dual_active_constraints}
	For every $r_\ell(\epsilon) \in R_\ell(\epsilon)$, $I \in \mathcal{I}_f$, and $a \in A(I)$, if Constraint~\eqref{eq:constraints_rewritten} holds with equality in an optimal solution to Problem~$\mathcal{D}(\epsilon)$, then
	\begin{align}\label{eq:active_cons_condition}
		\max_{\substack{\hat{r}_f \in R_f(a)}} g_{f, I}(r_\ell(\epsilon), \hat{r}_f) = \max_{\substack{\hat{r}_f \in R_f(I)}} g_{f, I}(r_\ell(\epsilon), \hat{r}_f).
	\end{align}
\end{restatable}

Now we are ready to prove a fundamental property of the follower's best responses in $\xi$-perturbed game instances $\Gamma(\epsilon)$.
Intuitively, in a perturbed game instance, the follower best responds playing sequence $\sigma(I_f)a$ with probability strictly greater than its lower-bound $\xi_f(\epsilon,\sigma_f(I)a)$ only if playing $a$ is optimal in the game following $I$.
%
%
Theorem~\ref{thm:compl_slack} formally expresses the idea that, in a perturbed game instance $\Gamma(\epsilon)$, when the follower decides how to best respond to a leader's commitment in a given information set, she does not take into account her future trembles, but only opponents' ones.
\begin{theorem}\label{thm:compl_slack}
	Given $r_\ell(\epsilon) \in R_\ell(\epsilon)$, $r_f(\epsilon) \in \mathsf{BR}_{\Gamma(\epsilon)}(r_\ell(\epsilon))$, $I \in \mathcal{I}_f$, and $a \in A(I)$,
if $r_f(\epsilon, \sigma_f(I) a) > \xi_f(\epsilon, \sigma_f(I) a)$, then
	$$\max_{\substack{\hat{r}_f \in R_f(a)}}  g_{f, I}(r_\ell(\epsilon) , \hat{r}_f) = \max_{\substack{\hat{r}_f \in R_f(I)}} g_{f, I}(r_\ell(\epsilon), \hat{r}_f).
	$$
\end{theorem}
\begin{proof}
By Lemma~\ref{lem:ne_pert_game}, $r_f(\epsilon) \in \mathsf{BR}_{\Gamma(\epsilon)}(r_\ell(\epsilon))$ if and only if $\tilde{r}_f(\epsilon) = r_f(\epsilon) - \xi_f(\epsilon)$ is optimal for Problem~$\mathcal{P}(\epsilon)$.
	%
	%
	%
	By applying the complementarity slackness theorem to Problems~$\mathcal{P}(\epsilon)$~and~$\mathcal{D}(\epsilon)$ we have that, if $\tilde{r}_f(\epsilon)$ and $v_f^\ast \in \mathbb{R}^{|\mathcal{I}_f|+1}$ are optimal, then,
	whenever $\tilde{r}_f(\epsilon, \sigma_f(I)a) > 0$, that is, $r_f(\epsilon, \sigma_f(I) a) > \xi_f(\epsilon, \sigma_f(I) a)$, Constraint~\eqref{eq:constraints_rewritten} for information set $I$ and action $a$ must hold with equality, which, by Lemma~\ref{lem:dual_active_constraints}, yields Equation~\eqref{eq:active_cons_condition}.
	%
\end{proof}

\section{Limits of SEs in $\xi$-Perturbed Games are QPSEs of the Unperturbed Games}\label{sec:limits_sse}

Here, we prove Theorem~\ref{thm:limit_se}.
%
%
%
%
First, we introduce two lemmas.

The first provides a characterization of $I$-best responses in terms of sequence form.
Intuitively, a follower's strategy $\pi_f$ is an $I$-best response to $\pi_\ell$ if and only if it places positive probability only on actions $a \in A(I)$ that are part of some best response of the follower below information set $I$.
\begin{restatable}{lemma}{lemmabr}\label{lem:lemma_br_I}
	Given an SEFG $\Gamma$, a completely mixed $\pi_\ell \in \Pi_\ell$ and $I \in \mathcal{I}_f$, $\pi_f \in \mathsf{BR}_{I}(\pi_\ell)$ if for every $a \in A(I)$:
	$$
	\pi_{i a} > 0 \implies \hspace{-0.3cm}\max_{\substack{\hat{r}_f \in R_f(a)}} \hspace{-0.1cm} g_{f, I}(r_\ell, \hat{r}_f) = \hspace{-0.3cm}\max_{\substack{\hat{r}_f \in R_f(I)}}  g_{f, I}(r_\ell, \hat{r}_f),
	$$
	where $r_\ell \in R_\ell$ is equivalent to $\pi_\ell$.
\end{restatable}


%
%
The next lemma shows that any limit point of a sequence of follower's best responses in $\xi$-perturbed games is a follower's best response at every information set in $\Gamma$.
\begin{restatable}{lemma}{lemfive}\label{lem:limit_follower_response}
	Given a $\xi$-perturbed SEFG $(\Gamma, \xi_\ell, \xi_f)$, let $\{ \epsilon_k \}_{k \in \mathbb{N}} \rightarrow 0$ and let $\{ (r_\ell(\epsilon_k), r_f(\epsilon_k)) \}_{k \in \mathbb{N}}$ be a sequence of realization plans in $\Gamma(\epsilon_k)$ with $r_f(\epsilon_k) \in \mathsf{BR}_{\Gamma(\epsilon_k)} (r_\ell(\epsilon_k))$.
	Then, any limit point $(\pi_\ell, \pi_f)$ of $\{ (\pi_{\ell, k}, \pi_{f, k}) \}_{k \in \mathbb{N}}$ is such that, eventually, $\pi_f \in \mathsf{BR}_{I_f}(\pi_{\ell,k})$ for all $I \in \mathcal{I}_f$,
	where $(\pi_{\ell, k}, \pi_{f, k})$ are equivalent to $(r_\ell(\epsilon_k), r_f(\epsilon_k))$ for all $k \in \mathbb{N}$.
\end{restatable}


Finally, we can prove Theorem~\ref{thm:limit_se}.

\begin{proof}[Proof of Theorem~\ref{thm:limit_se}]
First, since $r_f(\epsilon_k) \in \mathsf{BR}_{\Gamma(\epsilon_k)} (r_\ell(\epsilon_k))$ for all $k \in \mathbb{N}$, Lemma~\ref{lem:limit_follower_response} allows us to conclude that requirement (i) in Definition~\ref{def:qpse} holds.
Therefore, in order to prove Theorem~\ref{thm:limit_se}, we need to show that requirement (ii) holds as well.
For contradiction, suppose that point (ii) does not hold, that is, no matter how we choose sequences $\{ \pi_{i,k} \}_{k \in \mathbb{N}}$, for $i \in \mathcal{N}$ and $\pi_i \in \Pi_i$, there is an information set ${I} \in \mathcal{I}_\ell \cup \{ I_\emptyset \}$ and a leader's strategy $\hat{\pi}_\ell \in \Pi_\ell$ such that, for every $ \hat{\pi}_f \in \Pi_f : \hat{\pi}_f \in \mathsf{BR}_{\hat{I}}(\pi_{\ell,k} \big/_{I} \hat{\pi}_{\ell,k}) $ for all $ \hat{I} \in \mathcal{I}_f$, we have
$
u_{\ell} (\pi_{\ell,k} \big/_{I} \pi_\ell, \pi_{f,k} ) < u_{\ell} (\pi_{\ell,k} \big/_{I} \hat{\pi}_\ell, \hat{\pi}_{f,k} )
$.
By continuity, there must exist $\bar k \in \mathbb{N}$ such that, for all $k \in \mathbb{N}: k \geq \bar k$,
$
u_{\ell} (\pi_{\ell,k} \big/_{I} \pi_{\ell,k}, \pi_{f,k} )= u_{\ell} (\pi_{\ell,k}, \pi_{f,k} )  < u_{\ell} (\pi_{\ell,k} \big/_{I} \hat{\pi}_{\ell,k}, \hat{\pi}_{f,k} ).
$
Let sequence $\{ \hat{\pi}_{\ell, k} \}_{k \in \mathbb{N}}$ be such that $\hat{r}_\ell(\epsilon_k) \in R_\ell(\epsilon_k)$ for all $k \in \mathbb{N}$, where each realization plan $\hat{r}_\ell(\epsilon_k)$ is equivalent to the strategy $\pi_{\ell,k} \big/_{I} \hat{\pi}_{\ell,k}$.
This is always possible since requirement (iii) in Definition~\ref{def:qp_pert} is satisfied.
Consider a sequence $\{ (\hat{r}_\ell(\epsilon_k),  \hat{r}_f(\epsilon_k) \}_{k \in \mathbb{N}}$ with $\hat{r}_f(\epsilon_k) \in \mathsf{BR}_{\Gamma(\epsilon_k)}(\hat{r}_\ell(\epsilon_k))$, and let $\{  (\pi_{\ell,k} \big/_{I} \hat{\pi}_{\ell,k}, \hat{\pi}_{f, k}) \}_{k \in \mathbb{N}}$ be a sequence such that each strategy $\hat{\pi}_{f,k}$ is equivalent to $\hat{r}_f(\epsilon_k)$.
By Lemma~\ref{lem:limit_follower_response}, any limit point $(\pi_\ell \big/_{I} \hat{\pi}_\ell, \hat{\pi}_f)$ of $\{  (\pi_{\ell,k} \big/_{I} \hat{\pi}_{\ell,k}, \hat{\pi}_{f, k}) \}_{k \in \mathbb{N}}$ satisfies $\hat{\pi}_f \in \mathsf{BR}_{\hat{I}}(\pi_{\ell,k} \big/_{I} \hat{\pi}_{\ell,k}) $ for all $ \hat{I} \in \mathcal{I}_f$.
Thus, using the equivalence between strategies and realization plans, for $k \in \mathbb{N} : k \geq \bar k$ we have that $u_\ell(r_\ell(\epsilon_k), r_f(\epsilon_k)) < u_\ell(\hat{r}_\ell(\epsilon_k), \hat{r}_f(\epsilon_k))$, no matter how we choose $\hat{r}_f(\epsilon_k) \in \mathsf{BR}_{\Gamma(\epsilon_k)}(\hat{r}_\ell(\epsilon_k))$.
This contradicts the fact that $(r_\ell(\epsilon_k), r_f(\epsilon_k))$ is an SE of $\Gamma(\epsilon_k)$.
\end{proof}
%

\section{Algorithm}\label{sec:sefce}

One can use our perturbation scheme to compute an (approximate) QPSE. We
do this by developing an LP for computing an SEFCE in a given $\xi$-perturbed game instance, where we maximize the leader's value. We
then conduct a \emph{branch-and-bound} search on this SEFCE LP. It branches on
which actions to \emph{force} be recommended to the follower (by the correlation device of the SEFCE). The idea is
that, as long as we only recommend a single action to the follower at any given
information set, we get an SE of the perturbed game (specifically an SSE, and an SSE has maximum value among all SEs), and, thus, according to Theorem~\ref{thm:limit_se}, a QPSE (specifically QPSSE) if we take
the limit point of the perturbations. As in prior papers on EFCE
computation in general-sum games, we focus on games without chance nodes~\citep{Stengel08:Extensive,Cermak16:Using}.

For computing an SEFCE we need to specify joint probabilities over sequence
pairs $(\sigma_\ell,\sigma_f)\in \Sigma$. However, not all pairs
need to specify probabilities, only pairs such that choosing $\sigma_f$ is
affected by the probability put on $\sigma_\ell$ (we do not need to care
about the converse of this, as only the follower needs to be induced to follow
the recommended strategy). Intuitively, the set of the leader's sequences relevant to a
given $\sigma_f \in \Sigma_f$ is made of those sequences that affect the expected value
of $\sigma_f$ or any alternative sequence $\hat{\sigma}_f \in \Sigma_f$ whose last action is available at $I_f(\sigma_f)$.
\begin{definition}[Relevant sequences]
  A pair $(\sigma_\ell,\sigma_f) \in \Sigma$ is \emph{relevant} if either $\sigma_\ell = \sigma_\emptyset$ or there exists
  $h,\hat{h} \in \mathcal{H}$ s.t. $\hat{h}$ precedes $h$, 
  $h \in I_f(\sigma_f)$, and $\hat{h} \in I_\ell(\sigma_\ell)$, or if the condition
  holds with the roles of $\sigma_\ell$ and $\sigma_f$ reversed.
\end{definition}
For every information set $I\in \mathcal{I}_i$, we let $rel(I)$ be the set of sequences relevant to
each child sequence $\sigma_i(I)a$ for $a\in A(I)$.
We let $p(\sigma_\ell,\sigma_f)$ be the probability that we recommend that the
leader plays sequence $\sigma_\ell$, and that the follower sends her
\emph{residual} (\emph{i.e.}, the probability that is not fixed by the perturbation) to $\sigma_f$.
Moreover, we let $\eta(\sigma_f)$ be the maximum probability that the follower can put on
a sequence $\sigma_f \in \Sigma_f$ given the $\xi_f$-perturbation scheme.

First, we introduce a new value function representing the value to the leader of
the sequence pair $(\sigma_\ell,\sigma_f) \in \Sigma$ given that $\sigma_f$ represents an
assignment of residual probability:
\[
  u_\ell^{\epsilon}(\sigma_\ell,\sigma_f) = \sum_{\mathclap{\substack{h\in \leafnodes:
      \sigma_\ell(h) = \sigma_\ell \wedge \sigma_f(h) = \sigma_f}}} \eta(\sigma_f)u_\ell(h)
  + \sum_{\hat{\sigma}_f \in \Sigma_f} \xi_f(\epsilon, \hat{\sigma}_f) u_\ell(\sigma_\ell, \hat{\sigma}_f).
\]


The following LP finds an SEFCE in a $\xi$-perturbed SEFG.
\fontsize{8.5pt}{1pt}\selectfont
\begin{subequations}
\begin{align}
\max_{p,v}  & \quad 
                 \sum_{(\sigma_\ell,\sigma_f)\in \Sigma}  p(\sigma_\ell, \sigma_f)
                 u_\ell^{\epsilon}(\sigma_\ell,\sigma_f) \qquad \textrm{s.t.}
                 \\
&
                p(\emptyset, \emptyset)  =  1,  
\hspace{0.4cm}                p(\sigma_\ell,\sigma_f)     \geq 0 \hspace{1.7cm} \forall( \sigma_\ell, \sigma_f) \in \Sigma \\ 
&  \sum_{\sigma_f \in rel(\sigma_\ell)}p(\sigma_\ell,\sigma_f)     \geq \xi_\ell(\epsilon, \sigma_\ell) \hspace{1.91cm} \forall \sigma_\ell \in \Sigma_\ell \label{eq:xi_constraint}\\ 
&               p(\sigma_\ell(I),\sigma_f) = \sum_{\mathclap{a \in A(I)}} p(\sigma_\ell(I)a,\sigma_f) \hspace{0.31cm} \forall I\in \cI_\ell, \sigma_f \in rel(I) \label{eq:sefce_lp_leader_sequence_constraint}\\ 
&                p(\sigma_\ell,\sigma_f(I)) = \sum_{\mathclap{a \in A(I)}} p(\sigma_\ell,\sigma_f(I)a) \hspace{0.32cm} \forall I\in \cI_f, \sigma_\ell \in rel(I) \label{eq:sefce_lp_follower_sequence_constraint}\\
& v(\sigma_f) = \eta(\sigma_f) \sum_{\sigma_\ell \in rel(\sigma_f)} p(\sigma_\ell,\sigma_f) u_f(\sigma_\ell, \sigma_f) \label{eq:sefce_lp_incentive_sequence} \\
  & \hspace{1cm}+ \sum_{I\in \cI_f: \sigma_f(I)=\sigma_f}\sum_{a \in A(I)} v(\sigma_f a) \hspace{1.4cm}  \forall \sigma_f \in \Sigma_f \nonumber\\ 
  &                v(I,\sigma_f) \geq \eta(\sigma_f(I)a) \hspace{-2mm} \sum_{\sigma_\ell \in rel(\sigma_f)} p(\sigma_\ell,\sigma_f) u_f(\sigma_\ell, \sigma_f(I)a) \label{eq:sefce_lp_incentive_geq}\\
  &\hspace{0.8cm} + \,\, \sum_{\mathclap{\hat{I} \in \cI_f;\sigma_f(\hat{I})=\sigma_f(I)a}} v(\hat{I},\sigma_f) \hspace{0.9cm} 
    \forall I\in \cI_f, a \in A(I), \sigma_f \in prec(I)  \nonumber\\ 
&               v(\sigma_f(I)a) =  v(I,\sigma_f(I)a)  \hspace{1.45cm} \forall I\in \cI_f, a\in A(I).\label{eq:sefce_lp_incentive_optimality} 
\end{align}
\label{eq:eps sefce lp}
\end{subequations}
\normalsize%
In \eqref{eq:sefce_lp_incentive_geq} of this LP, $prec(I)$, where $I \in \mathcal{I}_f$, is the set of follower's
sequences $\sigma_f$ that precede $I$ in the sense that there is $\hat{I} \in \mathcal{I}_f$ with $\sigma_f(\hat{I})\sqsubseteq
\sigma_f(I)$ and $\sigma_f=\sigma_f(\hat{I})a$ for some $a\in A(\hat{I})$.
This LP is a modification of the SEFCE LP given
by~\citet{Cermak16:Using}. The new LP has two modifications to allow perturbation. First, it has constraints \eqref{eq:xi_constraint} to ensure that the sum of recommendation probabilities on any
leader's sequence is at least $\xi_\ell(\epsilon, \sigma_\ell)$. Second, because we are now
recommending where to send residual probability for the follower, we must modify
the objective in order to give the correct expected value for the leader.\footnote{We use the definition of relevant sequences and the LP from \citet{Stengel08:Extensive} rather than those of \citet{Cermak16:Using}. The latter are not well defined for \eqref{eq:sefce_lp_leader_sequence_constraint} and \eqref{eq:sefce_lp_follower_sequence_constraint}.}

We can branch-and-bound on recommendations to the follower in a way that ensures that the final outcome is an SSE.
That is guaranteed by the following theorem, which shows that we can add and
remove constraints on which follower actions to
recommend in a way that guarantees
an SSE of the perturbed game as long as the follower is recommended a ``pure''
strategy with respect to the residual probabilities.
\begin{restatable}{theorem}{thmsefcepure}
  If a solution to LP~\eqref{eq:eps sefce lp} is such that for all $I \in \cI_f$
  there exists $a\in A(I)$ such that $p(\sigma_\ell,\sigma_f(I)\hat{a}) = 0$ for all $\hat{a}
  \in A(I), \sigma_\ell\in rel(\sigma_f(I)a)$ with $\hat{a}\ne a$, then a strategy
  profile can be extracted in polynomial time such that it is an SSE of the perturbed game instance.
  \label{thm:sefce_pure}
\end{restatable}

Now it is obvious that the LP~\eqref{eq:eps sefce lp} upper bounds the value of any
SSE since the SSE is a feasible solution to the LP.

Theorem~\ref{thm:sefce_pure} shows that one way to find an SSE is to find a
solution to LP~\eqref{eq:eps sefce lp} where the follower is recommended a pure
strategy with respect to the residual probabilities. Since any SSE represents such a
solution, we can branch on which actions we make pure at each information set,
and use branch-and-bound to prune the space of possible solutions. This approach
was proposed by \citet{Cermak16:Using} for computing SSEs in unperturbed games,
where they showed that it performs better than a single MIP. Because our LP for
perturbed games uses residual probabilities for the follower, we can apply the
 branching methodology of \citet{Cermak16:Using}. At each node in the
search we choose some information set $I$ where more than one action is
recommended. We then branch on which action in $A(I)$ to recommend. Forcing a
given action is accomplished by requiring all other action probabilities be zero.
Our branch-and-bound chooses information sets according to depth, always branching on the shallowest one with at least two recommended action. We explore actions in
descending order of mass, where the mass on $a\in A(I)$ (with sequence
$\sigma_f$) is $\sum_{\sigma_\ell \in rel(\sigma_f)}p(\sigma_\ell,\sigma_f)$.

The algorithm finds an SSE of the perturbed game. In the limit as the perturbation approaches zero, this yields a QPSE. No algorithm is currently known for computing such an exact limit. In practice, we pick a small perturbation and solve the branch-and-bound using that value. This immediately leads to an approximate notion of QPSE (akin to approximate refinement notions in non-Stackelberg EFGs~\cite{farina2017regret,kroer2017smoothing}).
Another approach is to use our algorithm as an anytime algorithm where one runs it repeatedly with smaller and smaller perturbation values.


\section{Experiments}

We conducted experiments with our algorithm on two common benchmark EFGs. The
first is a search game played on the graph shown in
Figure~\ref{fig:search game}. It is a simultaneous-move game (which can be
modeled as a turn-taking EFG with appropriately chosen information sets). The
leader controls two patrols that can each move within their respective shaded
areas (labeled P1 and P2), and at each time step the controller chooses a move
for both patrols. The follower is always at a single node on the graph,
initially the leftmost node labeled $S$ and can move freely to any adjacent node
(except at patrolled nodes, the follower cannot move from a patrolled node to
another patrolled node). The follower can also choose to wait in place for a
time step in order to clean up their traces. If a patrol visits a node that was
previously visited by the follower, and the follower did not wait to clean up
their traces, they can see that the follower was there. If the follower reaches
any of the rightmost nodes they received the respective payoff at the node ($5$
and $10$, respectively). If the follower and any patrol are on the same node at
any time step, the follower is captured, which leads to a payoff of $0$ for the
follower and a payoff of $1$ for the leader. Finally, the game times out after
$k$ simultaneous moves, in which case the leader receives payoff $0$ and the
follower receives $-\infty$ (because we are interested in games where the
follower attempts to reach an end node). This is the game considered by
\citet{Kroer18:Robust} except with the bottom layer removed, and is similar to
games considered by \citet{Bosansky14:Exact} and \citet{Bosansky15:Sequence}.
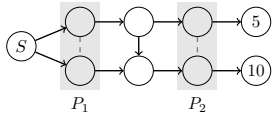
\begin{figure}[!h]
  \centering
  \scalebox{0.65}{
  \begin{tikzpicture}
\path[fill=black!10!white] (.8, -0.4) rectangle (1.6, 1.4);
\path[fill=black!10!white] (3.2, -0.4) rectangle (4.0, 1.4);
\node at (1.2, -0.7) {$P_1$};
\node at (3.6, -0.7) {$P_2$};

  \node[draw, circle, minimum width=.6cm, inner sep=0] (A) at (0, 0.5) {$S$};
  \node[draw, circle, minimum width=.6cm] (B) at (1.2, 1) {};
  \node[draw, circle, minimum width=.6cm] (C) at (1.2, 0) {};
    \node[draw, circle, minimum width=.6cm] (E) at (2.4, 1) {};
    \node[draw, circle, minimum width=.6cm] (F) at (2.4, 0) {};
      \node[draw, circle, minimum width=.6cm] (H) at (3.6, 1) {};
      \node[draw, circle, minimum width=.6cm] (I) at (3.6, 0) {};
            \node[draw, circle, minimum width=.6cm,inner sep=0] (K) at (4.8, 1) {$5$};
            \node[draw, circle, minimum width=.6cm,inner sep=0] (L) at (4.8, 0) {$10$};
            
\draw[thick,->] (A) edge (B);
\draw[thick,->] (A) edge (C);
\draw[thick,->] (B) edge (E);
\draw[thick,->] (C) edge (F);
\draw[thick,->] (E) edge (F);
\draw[thick,->] (E) edge (H);
\draw[thick,->] (F) edge (I);
\draw[thick,->] (H) edge (K);
\draw[thick,->] (I) edge (L);

\draw[thick,gray,dashed] (B) edge (C);

\draw[thick,gray,dashed] (H) edge (I);
\end{tikzpicture}
  }
  \caption{The graph on which the search game is played.}
  \label{fig:search game}
\end{figure}

The second game is a variant on Goofspiel~\citep{Ross71:Goofspiel}, a bidding
game where each player has a hand of cards numbered $1$ to $3$. There are $3$ prizes worth $1, \ldots, 3$.
In each turn, the prize is the smallest among the remaining prizes.
Within the turn, the each of two players simultaneously chooses some private card to play. The player with the larger card wins the prize. In case of a tie, the prize is discarded, so this is not a constant-sum game. The cards that were played get discarded.
Once all cards have been played, a player's score is the sum of the prizes that she has won.

The LP solver we used is GLPK 4.63~\cite{glpk4.63}. We had to make the following changes to GLPK.
First, we had to expose some internal routines so that we could input to the solver rational
numbers rather than double-precision numbers. Second, we fixed a glitch in GLPK's rational LP
solver in its pivoting step (it was not correct when the rational numbers were too small).
Our code and GLPK use the GNU GMP library to provide arbitrary-precision arithmetic. 
The code, written in the C++14 language, was compiled with the g++ 7.2.0 compiler. 
It was
run on a single thread on a 2.3 GHz Intel Xeon processor. The results are shown in Figure~\ref{fig:plot1}.
\begin{figure}[!ht]
  \centering
  \includegraphics[width=0.83\columnwidth]{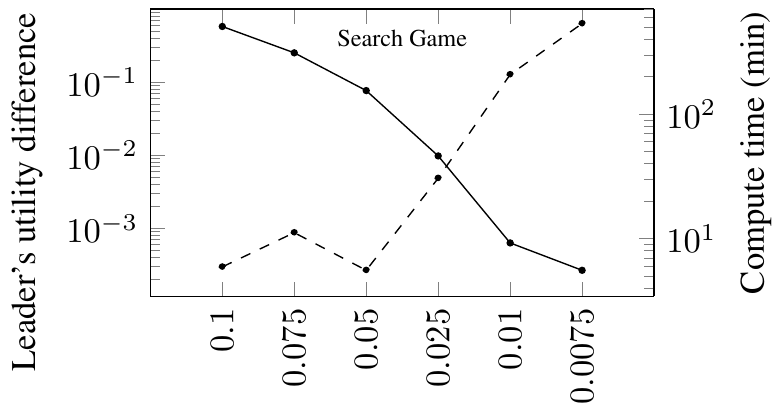}\\
  \includegraphics[width=0.83\columnwidth]{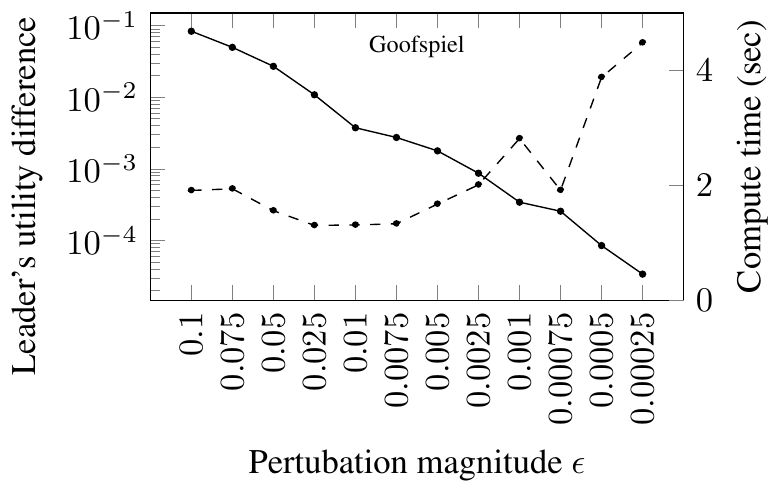}
  \caption{Experiments. Dashed lines show compute time. Solid lines show the loss in the leader's utility compared to the SSE value in the unperturbed game.}
  \label{fig:plot1}
\end{figure}


\section{Conclusions and Future Research}
  Quasi-perfect equilibrium has been studied in extensive-form games, but was poorly
  understood in Stackelberg settings.
  We provided a game-theoretic, axiomatic
  definition of \emph{quasi-perfect Stackelberg equilibrium (QPSE)}. We developed a family of
  game perturbation schemes that lead to a QPSE in the limit. Our
  family generalizes prior perturbation
  schemes introduced for finding (even non-Stackelberg) quasi-perfect
  equilibria. Using our perturbation schemes, we developed a branch-and-bound
  algorithm for QPSE. It leverages a perturbed variant of the linear program for computing a
  Stackelberg extensive-form correlated equilibrium. Experiments show that our algorithm can be used to find an approximate QPSE in games with thousands of nodes.

We showed that some perturbation schemes outside our family do not lead to QPSEs in some games.
It remains an open question whether our perturbation family fully characterizes the whole set of QPSEs.
As to requirement (i) in Definition~\ref{def:qp_pert}, can all QPSEs be captured by perturbation schemes that only use polynomial lower bounds on trembles?

It was recently shown that in non-Stackelberg extensive-form games, there exists a perturbation size that is small enough (while still strictly positive) that an exact refined (e.g., quasi-perfect) equilibrium can be found by solving a mathematical program with that perturbation size~\citep{miltersen2010computing,DBLP:conf/aaai/Farina017,Farina18:Practical}, and \citet{Farina18:Practical} provide an algorithm for checking whether a given guess of perturbation size is small enough. That obviates the need to try to explicitly compute a limit of a sequence. It would be interesting to see whether such theory can also be developed for Stackelberg extensive-form games---and for our perturbation family in particular. 

\bibliography{refs}
\bibliographystyle{named}

\clearpage

\section*{Appendix}

\subsection*{Omitted Proofs}

\lemmanepert*

\begin{proof}
	%
	%
	Since, $r_f(\epsilon) \in \mathsf{BR}_{\Gamma(\epsilon)}(r_\ell(\epsilon))$ if and only if
	$
	r_f(\epsilon) \in  \arg\max_{r_f: F_f r_f = f_f, r_f \geq \xi_f(\epsilon)} r_\ell(\epsilon)^T  U_f r_f
	$,
	%
	introducing variables $\tilde{r}_f = r_f - \xi_f(\epsilon)$ and dropping the constant term $r_\ell(\epsilon)^T U_i \xi_f(\epsilon) $ from the objective, we obtain that $r_f(\epsilon)$ must be an optimal solution to Problem~$\mathcal{P}(\epsilon)$.
	%
	%
	%
	%
\end{proof}

\lemmaperta*

\begin{proof}
	Let us consider Problem~$\mathcal{D}(\epsilon)$.
	First, observe that, for every information set $I \in \mathcal{I}_f$, the objective function coefficient for the variable $v_{f, I}$ is equal to $\xi_f(\epsilon,\sigma_f(I)) - \sum_{a \in A(I)} \xi_f(\epsilon,\sigma_f(I) a)$.
	Assuming $\Gamma(\epsilon)$ is well-defined, such coefficients are positive for every $v_{f, I}$.
	Then, in an optimal solution $v_f^\ast \in \mathbb{R}^{|\mathcal{I}_f|+1}$ to Problem~$\mathcal{D}(\epsilon)$, each variable $v_{f,I}$ is set to its minimum given Constraints~\eqref{eq:constraints_rewritten}.
	We prove Equation~\eqref{eq:inductive_argument} using a simple inductive argument.
	The base case of the induction is when there is no information set $\hat{I} \neq I \in \mathcal{I}_f$ with $I \succeq \hat{I}$.
	For every action $a \in A(I)$, $v_{f,I} \geq \sum_{\substack{\sigma \in \Sigma : \sigma_f = \sigma_f(I)a }} u_f(\sigma) r_\ell(\epsilon,\sigma_\ell)$, which, using the fact that $v_{f, I}^\ast$ must be set to its minimum possible value given the constraints, implies the following:
	\begin{align*}
	v_{f, I}^\ast &= \max_{a \in A(I)} \sum_{\substack{\sigma \in \Sigma : \sigma_f = \sigma_f(I) a }} u_f(\sigma) r_\ell(\epsilon,\sigma_\ell) = \\
	& = \max_{\substack{\hat{r}_f \in R_f(I)}} g_{f, I}( r_\ell(\epsilon), \hat{r}_f),
	\end{align*}
	where the last equality holds since $\sum_{a \in A(I)} \hat{r}_f(\sigma_f(I)a) = \hat{r}_f(\sigma_f(I)) =  1$, for the definition of realization plan.
	As for the inductive step, let us consider an information set $I \in \mathcal{I}_f$ and assume, by induction, that Equation~\eqref{eq:inductive_argument} holds for every information set $\hat{I} \neq I \in \mathcal{I}_f$ with $I \succeq \hat{I}$.
	We can write:
	\begin{align*}
	v_{f, I}^\ast &=  \max_{a \in A(I)} \hspace{-0.1cm} \sum_{\substack{\sigma \in \Sigma : \sigma_f =\sigma_f(I) a }} \hspace{-0.6cm} u_f(\sigma) r_\ell(\epsilon,\sigma_\ell) + \hspace{-0.8cm} \sum_{\substack{ \hat{I} \in \mathcal{I}_f: \sigma_f(\hat{I}) = \sigma_f(I)a}} \hspace{-0.8cm} v_{f,\hat{I}}^\ast  = \\
	& =  \max_{a \in A(I)} \sum_{\substack{\sigma \in \Sigma : \sigma_f = \sigma_f(I) a}}  u_f(\sigma) r_\ell(\epsilon,\sigma_\ell)  \quad+ \\
	& \quad\quad + \hspace{-0.5cm}  \sum_{\substack{ \hat{I} \in \mathcal{I}_f: \sigma_f(\hat{I}) = \sigma_f(I)a}} \max_{\substack{\hat{r}_f \in R_f(\hat{I})}}  g_{f, \hat{I}}(r_\ell(\epsilon), \hat{r}_f) = \\
	& = \max_{\substack{\hat{r}_f \in R_f(I)}} g_{f, I}(r_\ell(\epsilon), \hat{r}_f),
	\end{align*}
	where the first equality directly follows from the optimality of $v_f^\ast$, the second one from the inductive hypothesis, while the last equality holds since we have $\sum_{a \in A(I)} \hat{r}_f(\sigma_f(I)a) = \hat{r}_f(\sigma_f(I)) =  1$.
\end{proof}

\lemmapertb*

\begin{proof}
	Let $v^\ast \in \mathbb{R}^{|\mathcal{I}_i| + 1}$ be an optimal solution to Problem~$\mathcal{D}(\epsilon)$ that satisfies Constraint~\eqref{eq:constraints_rewritten}, for $I \in \mathcal{I}_f$ and $a \in A(I)$, with equality.
	We can write:
	\begin{align*}
	v_{f, I}^\ast &= \hspace{-0.3cm} \sum_{\substack{\sigma \in \Sigma : \sigma_f = \sigma_f(I) a}} \hspace{-0.5cm} u_f(\sigma) r_\ell(\epsilon,\sigma_\ell) + \hspace{-0.3cm} \sum_{\substack{ \hat{I}_f \in \mathcal{I}_f: \sigma_f(\hat{I}) = \sigma_f(I)a}} \hspace{-0.5cm} v_{f,\hat{I}}^\ast = \\
	& =\hspace{-0.3cm} \max_{\substack{\hat{r}_f \in R_f(a)}} g_{f,I}(r_\ell(\epsilon) , \hat{r}_f) = \hspace{-0.3cm}\max_{\substack{\hat{r}_f \in R_f(I)}} g_{f,I}(r_\ell(\epsilon), \hat{r}_f),
	\end{align*}
	where the second equality holds for the optimality of $v_f^\ast$ and the last one for Lemma~\ref{lem:optimal_dual}.
\end{proof}

\lemmabr*

\begin{proof}
	First, let us notice that, for every $I\in \mathcal{I}_f$ and $a \in A(I)$, the following relation holds:
	\begin{align}\label{eq:double_impl_agent_sequence}
	& \hspace{-0.2cm}  \max_{\substack{\hat{r}_f \in R_f(a)}}  g_{f, I}(r_\ell, \hat{r}_f) = \max_{\substack{\hat{r}_f \in R_f(I)}} g_{f,I}( r_\ell, \hat{r}_f ) \Longrightarrow \\
	&\hspace{-0.2cm} \max_{ \substack{\hat{\pi}_f \in \Pi_f : \hat{\pi}_{f a} = 1}} u_{f, I} \left( \pi_\ell, \pi_f \big/_{I} \hat{\pi}_f \right) = \max_{\hat{\pi}_f \in \Pi_f} u_{f, I} \left( \pi_\ell, \pi_f \big/_{I} \hat{\pi}_f \right)  \nonumber
	\end{align}
	In order to see this, for $I \in \mathcal{I}_f$, let $Z(I) \subseteq \mathcal{Z}$ be the set of terminal nodes that are potentially reachable from $I$, and, for $h \in Z(I)$ and $\hat{\pi}_f \in \Pi_f$, let $\mathcal{U}_{f,h}(\pi_\ell,  \hat{\pi}_f) = u_f(h) \prod_{a \in \sigma_\ell(h)} \pi_{\ell a}   \prod_{a \in \sigma_f(h) \setminus \sigma_f(I)  } \hat{\pi}_{fa} $.
	Given the realization equivalence of $r_\ell$ and $\pi_\ell$, and the fact that $\hat{r}_f(\sigma_f(I)) = 1$, the left-hand side in the first line of Equation~\eqref{eq:double_impl_agent_sequence} is equivalent to $\max_{\hat{\pi}_f \in \Pi_f : \hat{\pi}_{f a} = 1} \sum_{h \in Z(I)} \mathcal{U}_{f,h}(\pi_\ell,  \hat{\pi}_f) $, while the right-hand side is the same as $ \max_{\hat{\pi}_f \in \Pi_f} \sum_{h \in Z(I)} \mathcal{U}_{f,h}(\pi_\ell,  \hat{\pi}_f) $.
	Then, by dividing both sides of the equality in the first line of Equation~\eqref{eq:double_impl_agent_sequence} by $\sum_{h \in Z(I) }  \prod_{a \in \sigma_f(h)} \pi_{f a}$, by definition of $u_{f, I}(\pi_\ell,\pi_f \big/_{I} \hat{\pi}_f)$ we get the second line.
	%
	%
	Now, say that the condition of the lemma holds for every $a \in A(I)$.
	Clearly, we have $\max_{\hat{\pi}_f \in \Pi_f : \pi_f {=}_{I} \hat{\pi}_f  } u_{f,I} ( \pi_\ell,  \pi_f \big/_{I} \hat{\pi}_f ) = \sum_{a \in A(I)} \pi_{f a} \max_{\hat{\pi}_f \in \Pi_f : \hat{\pi}_{fa} = 1  } u_{f, I} ( \pi_\ell , \pi_f \big/_{I} \hat{\pi}_f  )  $, and, since $\pi_{fa} > 0$ only if $\max_{\substack{\hat{r}_f \in R_f(a)}} g_{f, I}(r_\ell, \hat{r}_f) = \max_{\substack{\hat{r}_f \in R_f(I)}} g_{f,I}(r_\ell,  \hat{r}_f)$, Eq.~\eqref{eq:double_impl_agent_sequence} proves the result.
	%
	%
\end{proof}

\lemfive*
\begin{proof}
	%
	%
	%
	First, notice that there must exist $\bar k \in \mathbb{N}$ such that, for all $k \in \mathbb{N} : k \geq \bar k$, and for every follower's information set $I \in \mathcal{I}_f$ and action $a \in A(I)$, if $\pi_{f a} > 0$, then ${r}_f(\epsilon_k,\sigma_f(I) a) > \xi_f(\epsilon_k,\sigma_f(I) a)$.
	Otherwise, by conditions (ii)-(iii) in Definition~\ref{def:qp_pert}, it would be $\pi_{fa} = 0$.
	Let us fix $I \in \mathcal{I}_f$ and $a \in A(I)$.
	Suppose that $\pi_{fa} > 0$.
	For all $k \in \mathbb{N} : k \geq \bar k $, we have that ${r}_f(\epsilon_k,\sigma_f(I) a) > \xi_f(\epsilon_k,\sigma_f(I) a)$, which, by Theorem~\ref{thm:compl_slack}, implies the following:
	$$
	\max_{\substack{\hat{r}_f \in R_f(a)}} g_{f, I}({r}_\ell(\epsilon_k) , \hat{r}_f) = \max_{\substack{\hat{r}_f \in R_f(I)}}  g_{f, I}({r}_\ell(\epsilon_k), \hat{r}_f).
	$$
	Thus, Lemma~\ref{lem:lemma_br_I} allows us to conclude that $\pi_f \in \mathsf{BR}_{I}(\pi_{\ell, k})$ for all $k \in \mathbb{N} : k \geq \bar k$, which proves the result.
	%
	%
	%
\end{proof}

\thmsefcepure*
\begin{proof}
  First, we check that the leader strategy is valid. The argument is
  identical to that of \citet{Cermak16:Using}. For the leader strategy at a
  given information set $I$ we pick an arbitrary $\sigma_f \in rel(\sigma_\ell(I))$
  that is played with positive probability and use the value
  $p(\sigma_\ell(I)a,\sigma_f)$ for all $a \in I$. All $\sigma_f \in
  rel(\sigma_\ell(I))$ recommend identical probability on $\sigma_\ell(I)a$ due to
  \eqref{eq:sefce_lp_leader_sequence_constraint} and the fact that we allow only
  a single follower action to be recommended at every follower information set.
  The incentive constraints \eqref{eq:sefce_lp_incentive_sequence} -
  \eqref{eq:sefce_lp_incentive_optimality} are identical to the original
  constraints given by \citet{Stengel08:Extensive}, so we only need to argue
  that we correctly represent the value of sending the residual along each
  sequence. But the value of sending the residual on $\sigma_f$ is simply the
  original value $\sum_{\sigma_\ell \in rel(\sigma_f)} p(\sigma_\ell,\sigma_f)
  u_f(\sigma_\ell, \sigma_f)$, except that we can send at most $\eta(\sigma_f)$
  probability on $\sigma_f$, plus the value of whichever choice we make for
  sending residual along descendants of $\sigma_f$. This is exactly the value
  that we encode in our constraints.
  It is easy to see that any SSE is a feasible solution to the LP: since the
  follower plays a pure strategy we can assign them their pure strategy, and
  assign the leader SSE strategy the same way across all follower
  recommendations.
\end{proof}

\subsection*{Limits of SSEs are QPSSEs}

Here, we show that limits of SSEs of perturbed SEFGs are QPSSEs of the original, unperturbed SEFGs, as the magnitude of the perturbation vanishes.

\begin{theorem}\label{thm:limit_sse}
	Given a perturbed SEFG $(\Gamma, \xi_\ell, \xi_f)$, let $\{ \epsilon_k \}_{k \in \mathbb{N}} \rightarrow 0$ and let $\{ (r_\ell(\epsilon_k), r_f(\epsilon_k)) \}_{k \in \mathbb{N}}$ be a sequence of SSEs in $\Gamma(\epsilon_k)$.
	Then, any limit point $(\pi_\ell, \pi_f)$ of the sequence $\{ (\pi_{\ell, k}, \pi_{f, k}) \}_{k \in \mathbb{N}}$ is a QPSSE of $\Gamma$,
	where $(\pi_{\ell, k}, \pi_{f, k})$ are equivalent to $(r_\ell(\epsilon_k), r_f(\epsilon_k))$ for all $k \in \mathbb{N}$.
\end{theorem}

\begin{proof}
		%
		First, as for Theorem~\ref{thm:limit_se}, Lemma~\ref{lem:limit_follower_response} allows us to conclude that point (i) in Definition~\ref{def:qpsse} holds.
		Let us prove point (ii).
		%
		%
		By contradiction, suppose that it does not hold,
		i.e., no matter how we choose sequences $\{ \pi_{i,k} \}_{k \in \mathbb{N}}$, for $i \in \mathcal{N}$ and $\pi_i \in \Pi_i$, there are ${I} \in \mathcal{I}_\ell \cup \{ I_\emptyset \}$, $\hat{\pi}_\ell \in \Pi_\ell$, and $ \hat{\pi}_f \in \Pi_f : \hat{\pi}_f \in \mathsf{BR}_{\hat I}(\pi_{\ell,k} \big/_{I} \hat{\pi}_{\ell,k}) $ for all $ \hat I \in \mathcal{I}_f$, with
		$
		u_{\ell} (\pi_{\ell,k} \big/_{I} \pi_\ell, \pi_{f,k} ) < u_{\ell} (\pi_{\ell,k} \big/_{I} \hat{\pi}_\ell, \hat{\pi}_{f,k} )
		$.
		%
		By continuity, there exists $\bar k \in \mathbb{N}$ such that, for all $k \in \mathbb{N}: k \geq \bar k$,
		$
		u_{\ell} (\pi_{\ell,k} \big/_{I} \pi_{\ell,k}, \pi_{f,k} ) = u_{\ell} (\pi_{\ell,k}, \pi_{f,k} )  < u_{\ell} (\pi_{\ell,k} \big/_{I} \hat{\pi}_{\ell,k}, \hat{\pi}_{f,k} ).
		$
		Let sequence $\{ \hat{\pi}_{\ell, k} \}_{k \in \mathbb{N}}$ be such that $\hat{r}_\ell(\epsilon_k) \in R_\ell(\epsilon_k)$ for all $k \in \mathbb{N}$, where each realization plan $\hat{r}_\ell(\epsilon_k)$ is equivalent to the strategy $\pi_{\ell,k} \big/_{I} \hat{\pi}_{\ell,k}$.
		Similarly, let sequence $\{ \hat{\pi}_{f, k} \}_{k \in \mathbb{N}}$ be such that $\hat{r}_f(\epsilon_k) \in R_f(\epsilon_k)$ for all $k \in \mathbb{N}$, where each $\hat{r}_f(\epsilon_k)$ is equivalent to $\hat{\pi}_{f,k}$.
		Notice that we can always choose the two sequences as described above, since we enforced point (iii) in Definition~\ref{def:qp_pert}.
		Clearly, $\hat{r}_f(\epsilon_k) \in \mathsf{BR}_{\Gamma(\epsilon_k)}(\hat{r}_\ell(\epsilon_k))$, otherwise $ \hat{\pi}_f \notin \mathsf{BR}_{\hat I}(\pi_{\ell,k} \big/_{I} \hat{\pi}_{\ell,k}) $ for some $ \hat I \in \mathcal{I}_f$, a contradiction.
		Using the equivalence between strategies and realization plans, for $k \in \mathbb{N} : k \geq \bar k$ we have that $u_\ell(r_\ell(\epsilon_k), r_f(\epsilon_k)) < u_\ell(\hat{r}_\ell(\epsilon_k),\hat{r}_f(\epsilon_k))$, which contradicts the fact that $(r_\ell(\epsilon_k), r_f(\epsilon_k))$ is an SSE of $\Gamma(\epsilon_k)$.
\end{proof}

\end{document}